\documentclass[12pt,a4paper]{article}

\usepackage{amssymb}
\usepackage{amsmath, amsthm}
\usepackage{arydshln}
\usepackage{epsfig}
\usepackage{setspace}

\usepackage{geometry}

\def\RR{{\mathbb{R}}}

\numberwithin{equation}{section}

\newcommand{\Conv}{\mathop{\rm Conv} }

\newtheorem{Thm}{Theorem}[section]
\newtheorem{Prop}[Thm]{Proposition}
\newtheorem{Lem}[Thm]{Lemma}
\newtheorem{Cor}[Thm]{Corollary}
\theoremstyle{definition}

\newtheorem{Rem}[Thm]{Remark}

\title{Finding Hall blockers by matrix scaling}
\author{Koyo Hayashi\footnote{
			Department of Computer Science, Graduate School of Information Science
		and Technology, University of Tokyo,
		Hongo 7-3-1, Bunkyo-ku, Tokyo 113-8656, Japan.	
		\texttt{khayashi@is.s.u-tokyo.ac.jp}} \quad \quad
	Hiroshi Hirai\footnote{Corresponding author,  
		Graduate School of Mathematics,  
		Nagoya University, 
		Furocho, Chikusaku, Nagoya, 464-8602, Japan.
		\texttt{hirai.hiroshi@math.nagoya-u.ac.jp}} \quad \quad
	Keiya Sakabe\footnote{
	Department of Mathematical Informatics, 
	Graduate School of Information Science
	and Technology, University of Tokyo,
		Hongo 7-3-1, Bunkyo-ku, Tokyo 113-8656, Japan.
		\texttt{ksakabe@g.ecc.u-tokyo.ac.jp}}
}
\begin{document}
\maketitle
\begin{abstract}
	For a given nonnegative matrix $A=(A_{ij})$,  
	the matrix scaling problem asks whether $A$ can be scaled to a doubly stochastic matrix $D_1AD_2$ for some positive diagonal matrices $D_1,D_2$.
	The Sinkhorn algorithm is a simple iterative algorithm,  
	which repeats row-normalization $A_{ij} \leftarrow A_{ij}/\sum_{j}A_{ij}$ and column-normalization $A_{ij} \leftarrow A_{ij}/\sum_{i}A_{ij}$ alternatively.
	By this algorithm, $A$ converges to a doubly stochastic matrix in limit if and only if the bipartite graph associated with $A$ has a perfect matching.
	This property can decide the existence of a perfect matching in a given bipartite graph $G$, which is identified with the $0,1$-matrix $A_G$.
	Linial, Samorodnitsky, and Wigderson showed that 
	$O(n^2 \log n)$ iterations for $A_G$
	decide whether $G$ has a perfect matching.
	Here $n$ is the number of vertices in one of the color classes of $G$.
	
	In this paper, we show an extension of this result:
	If $G$ has no perfect matching, 
	then a polynomial number of the Sinkhorn iterations 
	identifies a Hall blocker---a vertex subset $X$ 
	having neighbors $\Gamma(X)$ with $|X| > |\Gamma(X)|$, which is
	a certificate of the nonexistence of a perfect matching.
	Specifically, we show that $O(n^2 \log n)$ iterations can identify one Hall blocker, 
	and that further polynomial iterations can also identify all parametric Hall blockers 
	$X$ of maximizing $(1-\lambda) |X| - \lambda |\Gamma(X)|$ for $\lambda \in [0,1]$.
	The former result is based on an interpretation of the Sinkhorn algorithm as 
	alternating minimization for geometric programming.    
	The latter is on an interpretation as  
	alternating minimization for KL-divergence
	(Csisz\'{a}r and Tusn\'{a}dy 1984,  Gietl and Reffel 2013)
	and its limiting behavior for a nonscalable matrix (Aas 2014).
	We also relate the Sinkhorn limit with
	parametric network flow, principal partition of polymatroids, 
	and the Dulmage-Mendelsohn decomposition of a bipartite graph.
\end{abstract}


\section{Introduction}
For a given nonnegative matrix $A = (A_{ij})$,  
the {\em (doubly stochastic) matrix scaling problem}~\cite{Sinkhorn64} 
asks whether $A$ can be scaled to a doubly stochastic matrix $D_1AD_2$ 
for some positive diagonal matrices $D_1,D_2$.
The {\em Sinkhorn algorithm} (also called {\em RAS algorithm} or {\em IFP procedure}) is a simple iterative algorithm,  
which repeats row normalization 
$A_{ij} \leftarrow A_{ij}/\sum_{j}A_{ij}$ and 
column normalization $A_{ij} \leftarrow A_{ij}/\sum_{i} A_{ij}$ alternatively.
See Idel's survey~\cite{Idel_survey}
for rich literature of matrix scaling and the Sinkhorn algorithm.
By this algorithm, $A$ converges to a doubly stochastic matrix in limit
if and only if 
the bipartite graph $G$ associated with $A$ has a perfect matching~\cite{SinkhornKnopp67}, where 
$G$ is defined as a bipartite graph having row and column indices of $A$ as vertices
and having edges $ij$ if $A_{ij} > 0$.
This convergence property can decide whether a given bipartite graph 
$G$ has a perfect matching. Consider 0,1-matrix 
$A_G$ defined by $(A_G)_{ij} := 1$ for $ij \in E$ and zero for other indices.
Linial, Samorodnitsky, and Wigderson~\cite{LGW00} showed that
a polynomial number of iterations of the Sinkhorn algorithm applied to $A = A_G$ 
decides whether $G$ has a perfect matching.
The same result is independently obtained by 
Gurvits and Yianilos~\cite{GurvitsYianilos1998}.
Specifically,  within $O(n^2 \log n)$ iterations,
if $\| {\bf 1} - A{\bf 1}\|_2^2 < 1/n$ holds after column normalization, 
then there is a perfect matching.
Otherwise, there is no perfect matching.
Here ${\bf 1}$ denotes the all-one vector, and
$n$ is the number of vertices in one of the color classes of $G$.
Although the time complexity of this algorithm is inferior 
than the standard augmenting path algorithm, 
it is interesting in its simplicity of the algorithm description
as well as its conceptual difference from the standard one.
Also it is linked with the recent development on operator scaling, which is mentioned in the last of introduction.

In this paper, we address an extension of this result.
Hall's marriage theorem says that $G$ has a perfect matching if and only 
if there is no {\em Hall blocker}---a subset $X$ of vertices in one color class
having neighbors $\Gamma(X)$ with $|X| > |\Gamma(X)|$.
The above algorithm outputs neither a perfect matching nor a Hall blocker.  
The main result of this paper is to show that 
a polynomial number of the Sinkhorn iterations
can identify Hall blockers if $G$ has no perfect matching.

The first result is that $O(n^2 \log n)$ iterations can identify 
a Hall blocker if $G$ has no perfect matching.  
The analysis is based on an interpretation of the Sinkhorn algorithm as 
an alternating minimization for {\em capacity} 
$\inf_{x_1,x_2,\ldots,x_n > 0} \prod_{i}\sum_j A_{ij}x_j/(x_1x_2\ldots x_n)$, 
as in the analysis of the above algorithm~\cite{GurvitsYianilos1998,LGW00}.
If $G$ has no perfect matching,
then the convex optimization ({\em geometric programming}) formulation of 
the logarithmic capacity has unbounded optimal value $-\infty$.
By a simple criterion for unbounded geometric programs,
together with decrement analysis of one Sinkhorn iteration, 
we show that $O(n^2 \log n)$ iterations make scaling matrices $D_1,D_2$ satisfy the unboundedness criterion.
Then,  by sorting diagonals of $D_1,D_2$, we obtain a Hall blocker. 
This recovering procedure of a Hall blocker from $D_1,D_2$
is given by Franks, Soma, Goemans~\cite{FGS2022} in 
a warming-up argument for general setting of operator scaling.  
The new point here is the iteration bound, where we use a capacity lower bound specific to matrix scaling, 
not amenable to operator scaling.
The obtained Hall blocker $X$ has no guarantee that it maximizes violation $|X| - |\Gamma(X)|$. 

The second result is that further polynomial iterations
can identify a Hall blocker $X$ of maximum violation $|X| - |\Gamma(X)|$, 
more precisely, parametric Hall blockers of maximizing $(1-\lambda)|X| - \lambda |\Gamma(X)|$ for $\lambda \in [0,1]$.
The analysis is based on another alternating minimization interpretation of the Sinkhorn algorithm.
That is, it is also viewed as alternating minimization for
the minimum KL-divergence between the row-normalized and column-normalized spaces of matrices.
This alternating minimization formulation and its convergence property
can be analyzed via the information-geometric framework by Csisz\'{a}r and Tusn\'{a}dy~\cite{CsiszarTusnady1984}.
This fact was pointed out by Gietl and Reffel~\cite{GietlReffel}.
When $A$ is not scalable, 
the Sinkhorn algorithm does not converge.
However, it becomes oscillating between a pair of matrices attaining 
the minimum KL-divergence~\cite{GietlReffel}.
Then Aas~\cite{Aas2014} revealed a block-diagonalized structure of the oscillating limit, 
and gave a polynomial time procedure to determine to this block structure.

It turns out that this block structure includes information of parametric Hall blockers.
We present refined analysis on the oscillating limit 
according to the idea of {\em principal partition of polymatroids}~\cite{Fujishige2009,TomizawaFujishige}.
Particularly, we reveal that the block structure of the limit matrix coincides with  
a refined version of the {\em Dulmage-Mendelsohn decomposition}~\cite{DM1},
and is obtained via parametric network flow. 
Also we give an explicit formula of the limit of 
the row-marginal $A^*{\bf 1}$ of the column-normalized limit $A^*$.
The block-diagonal form is obtained by arranging $A$ 
with respect to the decreasing order of $A^*{\bf 1}$.
We provide an estimate of convergence to $A^*{\bf 1}$.
After a polynomial number of Sinkhorn iterations, 
the row-marginal $A{\bf 1}$ after column-normalization is sufficiently close to $A{\bf 1}^*$.
By sorting marginal $A{\bf 1}$, we can identify Hall blockers $X$ with maximum 
$(1-\lambda) |X| -  \lambda|\Gamma(X)|$ for $\lambda \in [0,1]$.
They include a Hall blocker with maximum violation $|X| - |\Gamma(X)|$.
The estimated number $O(n^6 \log n)$ of required iterations 
has a large gap to the above $O(n^2 \log n)$.
This may be caused from our primitive convergence analysis. 
The convergence rate of the Sinkhorn algorithm for scalable case  is well-studied; see \cite[Section 7]{Idel_survey}.
We could not find any existing work on convergence analysis 
of the Sinkhorn algorithm for nonscalable case. 
We expect that further analysis 
will improve the iteration bound, and close the gap.

It is an interesting future research to extend 
the present results to {\em operator scaling}~\cite{Franks2018, GGOW,OperatorScaling_survey,Gurvits04}---
a quantum generalization of matrix scaling to completely positive operators.
The Sinkhorn algorithm is generalized to the {\em operator Sinkhorn algorithm} 
(or {\em Gurvits algorithm})~\cite{Gurvits04}.
A recent important result~\cite{GGOW} is that this algorithm can decide 
the (doubly stochastic) scalability in polynomial time, 
which is viewed as a generalization of 
the above-mentioned perfect matching testing by matrix scaling. 
A Hall blocker, a certificate for nonscalability in this setting, 
is a certain vector subspace (called a {\em shrunk subspace}).
It is a natural question whether 
the operator Sinkhorn algorithm can find a Hall blocker for nonscalable case. 
This may need to characterize the limits of the operator Sinkhorn algorithm, 
which is raised as an open problem in \cite{OperatorScaling_survey}.
Although the operator Sinkhorn algorithm is 
interpreted as alternating minimization for the capacity in this setting,      
a divergence minimization interpretation is unknown; 
see \cite{MatsudaSoma2020} for a progress in this direction.
We hope that our results will be useful for tackling these problems. 

After we finished the first version of this paper, 
Franks, Soma, and Goemans~\cite{FGS2022} succeeded,  with avoiding these problems, to
modify
the operator Sinkhorn algorithm so that 
it obtains a shrunk subspace in polynomial time.
In \cite[Section 3.1]{FGS2022}, they presented the matrix scaling specialization of 
their modified algorithm, and gave a procedure of constructing a Hall blocker 
from their modified geometric programming formulation.
In the revision, 
we examine and adapt their argument for the original formulation, and
obtained the $O(n^2 \log n)$ bound by additional analysis of iteration complexity.

\paragraph{Organization.}
In Section~\ref{sec:matrixscaling}, 
we summarize basic facts on matrix scaling with emphasis on nonscalable case.
In Section~\ref{sec:limit}, 
we provide a polymatroid-based analysis on the limit of the Sinkhorn algorithm.
In Section~\ref{sec:finding},
we present our algorithms and prove the correctness.

\paragraph{Notation.}
For an positive integer $n$, let $[n] := \{1,2,\ldots,n\}$.
Let $\RR$ denote the set of real numbers.
For a vector $p \in \RR^n$ and subset $X \subseteq [n]$, let $p(X) := \sum_{i \in X}p_i$.
For a subset $X \subseteq [n]$, 
let ${\bf 1}_{X} \in \RR^n$ denote 
the 0,1-vector defined by $({\bf 1}_{X})_i=1 \Leftrightarrow i \in X$.
The all one vector ${\bf 1}_{[n]}$ is simply denoted by ${\bf 1}$.

For an $n \times m$ matrix $A  = (A_{ij}) \in \RR^{n \times m}$, 
$I \subseteq [n]$, and $J \subseteq [m]$, 
let $A[I,J]$ denote the submatrix of 
$A$ consisting of $A_{ij}$ for $i \in I, j \in J$.
If $A$ is a block diagonal matrix with diagonals $D_1,D_2,\ldots,D_k$, 
then $A$ is written as $A = \bigoplus_{i \in [k]} D_i$.

For two nonnegative vectors $p,q \in \RR^n$,
let $D(p \| q)$ denote the KL-divergence between $p$ and $q$:
\[
D(p \| q) := \sum_{i=1}^n p_i \log p_i/q_i,
\]
where $0 \log 0/x := 0$ for $x \geq 0$ and $x \log x/0 := \infty$ for $x > 0$.
In this paper, we allow $p,q$ to have different sum $p([n]) \neq q([n])$; so $D(p \| q)$ can be negative for such a case.

When $p([n]) = q([n])$, 
the KL-divergence is nonnegative, and relates to the $\ell_1$-distance via {\em Pinsker's inequality}: 
\begin{Lem}[{Pinsker's inequality; see~\cite[Lemma 11.6.1]{InformationTheory}}]\label{lem:Pinsker}
	For nonnegative vectors $p,q \in \RR^n$ with $p([n]) = q([n])$
\begin{equation}\label{eqn:Pinsker}
D(p \| q) \geq \frac{1}{2 p([n])} \|p-q\|_1^2. 
\end{equation}
\end{Lem}
Another useful property of the KL-divergence is Pythagorean theorem:
\begin{Lem}[{Pythagorean theorem; see~\cite[Theorem 11.6.1]{InformationTheory}}]\label{lem:Pythagorean}
Let $P$ be a closed (compact) convex set in $\{ p \in \RR^n \mid 
{\bf 1}^{\top} p = C, p_i \geq 0\ (i \in [n])\}$ for some $C \geq 0$.
For a nonnegative vector $r \in \RR^n$, let $p^* \in P$ satisfy $D(p^* \| r) = \inf_{p \in P} D(p \| r)$. Then it holds 
\[
D(p \| p^*) + D(p^* \| r) \leq D(p \| r) \quad (\forall p \in P).	
\]
\end{Lem}
Note that this is true for the case of  $C \neq r([n])$, since
$D(p \| r) = D(p \| (C/r([n]))r  ) - C \log C/r([n])$ for $p \in P$.


\section{Matrix scaling}\label{sec:matrixscaling}

In this section, we summarize basic results on matrix scaling 
with emphasis on nonscalable case.
Let $A =(A_{ij})$ be an $n \times m$ nonnegative matrix.
We assume that $A$ has neither zero rows nor  zero columns.
A {\em scaling} of $A$
is a matrix  $\tilde A = (\tilde A_{ij})$ written as 
$\tilde A_{ij} = A_{ij} x_i y_j$ for some positive vectors ({\em scaling vectors})
$x\in \RR^n, y \in \RR^{m}$.

Let $r$ and $c$ be $n$-dimensional and $m$-dimensional positive vectors, respectively.
An {\em $(r,c)$-scaling} of $A$ 
is a scaling $\tilde A$ of $A$ 
such that $\tilde{A}{\bf 1} = r$ and $\tilde{A}^{\top}{\bf 1} = c$.
The {\em matrix scaling problem} asks an $(r,c)$-scaling of $A$.
If an $(r,c)$-scaling of $A$ exists, then $A$ is said to be {\em $(r,c)$-scalable.}
If there is a scaling $\tilde A$ of $A$ such that 
$(\tilde{A}{\bf 1}, \tilde{A}^{\top}{\bf 1})$ is arbitrarily close $(r,c)$ 
(under some norm of $\RR^{n+m}$), then $A$ is {\em approximately $(r,c)$-scalable,}
where ``nonscalable" is meant as ``not approximately $(r,c)$-scalable."

Let $R := r([n])$ and $C := c([m])$. 
We do not assume $R=C$, which is an obvious necessary condition for scalability.
The case of $(r,c) = ({\bf 1}, {\bf 1})$ corresponds to the doubly stochastic scaling mentioned in introduction.

We explain combinatorial characterizations 
of scalability by Rothblum and Schneider~\cite{RothblumSchneider}.
Let $G(A) = ([n] \sqcup [m], E)$ denote the bipartite graph representing the nonzero pattern of $A$. 
Specifically, vertices $i \in [n]$ and $j \in [m]$ 
are joined by an edge $ij$ in $G(A)$
if and only if $A_{ij} > 0$. 
The notation $X \sqcup Y$ is meant as 
the disjoint union of $X \subseteq [n]$ and $Y \subseteq [m]$. 
A subset $X \sqcup Y$ of vertices is said to be {\em stable} 
if there is no edge $ij$ with $i\in X$ and $j \in Y$.
Let ${\cal S}$ denote the set of all stable set.
\begin{Thm}[\cite{RothblumSchneider}]\label{thm:RS}
	\begin{itemize}
		\item[(1)] $A$ is approximately $(r,c)$-scalable if and only if  $R=C$ and $r(X) + c(Y) \leq C$ for every $X \sqcup Y \in {\cal S}$.
		\item[(2)]  $A$ is $(r,c)$-scalable if and only if $R=C$, $r(X) + c(Y) \leq C$ for every $X \sqcup Y \in {\cal S}$, and 
		the equality $r(X) + c(Y) = C$ for $X \sqcup Y \in {\cal S}$ 
		implies $A[[n] \setminus X, [m] \setminus Y] = O$. 
		\item[(3)] An $(r,c)$-scaling of $A$ is unique if it exists. 
	\end{itemize} 
\end{Thm}

The only if part of (1) can be seen by the following estimate, 
which is an $\ell_1$-variant of \cite[Lemma 5.2]{LGW00}.
\begin{Lem}\label{lem:l_1}
	For every scaling $\tilde A$ of $A$, it holds
	$$
	\| \tilde A{\bf 1} -  r \|_1  + \| \tilde A^{\top}{\bf 1} - c \|_1 \geq 2 \max_{X \sqcup Y \in {\cal S}} r(X)+ c(Y) - (R+C)/2.
	$$
\end{Lem}
\begin{proof}
	Let $p := \tilde A{\bf 1}$, $q := \tilde A^{\top}{\bf 1}$, and $X \sqcup Y \in {\cal S}$.  
	Since $\tilde A_{ij} = 0$ for $i \in X, j \in Y$, 
	it holds that $\sum_{i \in X} p_i \leq \sum_{j \in [m] \setminus Y} q_j$ 
	and $\sum_{j \in Y} q_j \leq \sum_{i \in [n] \setminus X} p_i$.
	Then we have
	\begin{eqnarray*}
		&&	\| \tilde A{\bf 1} - r \|_1  + \| \tilde A^{\top}{\bf 1} - c \|_1  = \sum_{i=1}^n |p_i - r_i|+ \sum_{j=1}^m |q_j - c_j| \\
		&& =  \sum_{i \in X} |p_i - r_i|+ \sum_{j \in [m]\setminus Y}| q_j - c_j| +  \sum_{i \in [n] \setminus X} |p_i - r_i|+ \sum_{j \in Y}| q_j - c_j| \\
		&&\geq  r(X)- \sum_{i \in X} p_i  + \sum_{j \in [m]\setminus Y} q_j - c([m] \setminus Y) + \sum_{i \in [n] \setminus X} p_i - r([n] \setminus X)
		+ c(Y) - \sum_{j \in Y} q_j \\
		&& \geq 2 r(X) + 2 c(Y) - R-C.
	\end{eqnarray*}
\end{proof}
Let $\Gamma: 2^{[n]} \to 2^{[m]}$ be defined 
by 
\begin{equation}\label{eqn:Gamma}
	\Gamma(X) := \{j \in [m] \mid \mbox{$j$ is adjacent to a node in $X$ in $G(A)$}\}.
\end{equation}
Since every stable set $X \sqcup Y$ is contained in 
a larger stable set $X \sqcup ([m] \setminus \Gamma(X))$.
The condition in Theorem~\ref{thm:RS}~(1) is also written as $R = r([n]) = c([m]) = C$ and
$r(X) \leq c(\Gamma(X))$ for all $X \subseteq [n]$.
A subset $X$ violating this condition is a certificate of nonscalability, 
and is also called a {\em Hall blocker}.

%

\paragraph{Sinkhorn algorithm.}
The {\em Sinkhorn algorithm} is a simple method 
to obtain an approximate $(r,c)$-scaling.
The {\em row-normalization} $R(A)$ of $A$ is a scaling of $A$ 
defined by $R(A)_{ij} := (r_i/p_i) A_{ij}$ for $p_i := (A{\bf 1})_i$.
Similarly, the {\em column-normalization} $C(A)$ of $A$ is a scaling of $A$ 
defined by $C(A)_{ij} := (c_j/q_j) A_{ij}$ for $q_j := (A^{\top}{\bf 1})_j$.
The Sinkhorn algorithm is to repeat row-normalization $A \leftarrow R(A)$
and column-normalization $A \leftarrow C(A)$ alternatively.

This procedure is also described as update of scaling vectors $x,y$ of scaling $\tilde A= (A_{ij}x_iy_j)$.
The {\em row-normalization vector} $R(y) \in \RR^n$ relative to 
column-scaling vector $y$ is defined by 
$R(y)_i := r_i/\sum_{j} A_{ij}y_j$.
The {\em column-normalization vector} $C(x) \in \RR^m$ relative to row-scaling vector $x$ is defined by
$C(x)_j := c_j/\sum_{i} A_{ij}x_i$.
The Sinkhorn algorithm is to repeat row-normalization $x \leftarrow R(y)$
and column-normalization $y \leftarrow C(x)$ alternatively.


\subsection{Matrix scaling as geometric programming}\label{subsec:geometric}
In this section, we consider geometric programming formulation of matrix scaling, where 
we assume $R=C$.
It is well-known (see \cite{Idel_survey,RothblumSchneider}) that
the Sinkhorn algorithm is viewed as alternating minimization of 
the following optimization problem:
\begin{equation}\label{eqn:capacity}
	{\rm inf.}\quad \kappa(x,y) := \frac{\sum_{i,j} A_{ij}x_iy_j}{\prod_i x_i^{r_i/R}  \prod_j y_j^{c_j/C}} \quad {\rm s.t.}\quad  x_i,y_j > 0 \ (i \in [n],j \in [m]).
\end{equation} 
Fixing $y$, an optimal $x$ is given by row-normalization vector $R(y)$. 
Also, fixing $x$, an optimal $y$ is given by column-normalization vector $C(x)$. 
\begin{Thm}[\cite{RothblumSchneider}]\label{thm:scalability_bounded}
\begin{itemize}
	\item[(1)]  $A$ is $(r,c)$-scalable if and only if $\inf_{x,y} \kappa(x,y)$ is attained by some $x,y$. 
	\item[(2)] $A$ is approximately $(r,c)$-scalable if and only if $\inf_{x,y} \kappa(x,y) > 0$.
\end{itemize}
\end{Thm}
In the situation of (1), 
an $(r,c)$-scaling is given by $(A_{ij}x_iy_j)$.
In (2), an approximate $(r,c)$-scaling $(A_{ij}x_iy_j)$ is obtained by
near optimal solution $x,y$ with a prescribed accuracy.
 
The following formula of the decrement of $\kappa$ 
in row- and column-normalization is a variant of \cite[Lemma 2]{AWR2017}. 
\begin{Lem}~\label{lem:decrement} Let $y = C(x)$, $x' := R(y)$, and $y' := C(x')$. Then it holds
	\begin{equation}\label{eqn:decrement}
	\log \kappa(x',y') - \log \kappa(x,y) = - D\left( r \| p \right)/R 
	- D( c \| q )/C, 
	\end{equation}
	where $p_i := \sum_{j} A_{ij}x_iy_j$ and 
	$q_j := \sum_{i} A_{ij}x'_iy_j$. 
\end{Lem}
\begin{proof}
	Note that $x'$ is written as $x'_i = r_i x_i/p_i$.
	By calculation, we have
	$\log \kappa(x',y) - \log \kappa(x,y) = - D( r \| p )/R$.
	 Similarly, we have $\log \kappa(x',y') - \log \kappa(x',y) = - D( c \| q )/C$.
\end{proof}

\paragraph{Unbounded certificate in geometric programming.}
The problem (\ref{eqn:capacity}) is transformed into convex optimization 
by taking logarithm with variable change $x_i = e^{\xi_i}$, $y_j= e^{\eta_j}$:
\begin{equation}\label{eqn:capacity_geometric}
	\inf. \quad \log \sum_{i,j} A_{ij}e^{\xi_i+ \eta_j} - r^{\top}\xi/R - c^{\top}\eta/C \quad {\rm s.t.}\quad (\xi,\eta) \in \RR^n \times \RR^m.
\end{equation}
This problem falls into the class of convex optimization,  
called {\em geometric programming}~\cite{BJVH2007}.
A general form of unconstraint geometric program is:
Given positive reals $q_k > 0$ and vectors $\omega_k \in \RR^d$ $(k=1,2,\ldots,m \geq 2)$, 
minimize  function $f:\RR^d \to \RR$ defined by
\begin{equation}
	f(\xi) := \log \sum_{k=1}^m q_k e^{\omega_k^{\top} \xi} \quad (\xi \in \RR^d).
\end{equation}
Then, 
our problem~(\ref{eqn:capacity_geometric}) is for the case of $q_{ij} := A_{ij}$ and $\omega_{ij} := (e_i,e_j) - (r/R,c/C)$ for $i,j$ with $A_{ij} > 0$.
Theorem~\ref{thm:scalability_bounded} (2) says that
the approximate scalability is equivalent to the geometric program~(\ref{eqn:capacity_geometric}) being bounded below. 

We observe a simple criterion for 
(un)boundedness of unconstraint geometric program:
\begin{Lem}\label{lem:certificate_geometric}
	\begin{itemize}
		\item[(1)] $f$ is unbounded below if and only if 
		there is $\xi \in \RR^d$ such that $\omega_k^{\top}\xi < 0$ holds for every $k$.
		\item[(2)] Let $q_{\rm min} := \min_{k} q_k > 0$. If $f(\xi) \leq \log q_{\rm min}$, 
		then $\omega_k^{\top}\xi < 0$ holds for every $k$.
	\end{itemize}
\end{Lem}
\begin{proof}
	(2). If  $\omega_{k'}^{\top}\xi \geq 0$ for some $k'$, 
	then $f(\xi) = \log \sum_{k} q_k e^{\omega_k^{\top}\xi} > \log q_{k'} \geq \log q_{\rm min}$. 	
	
	(1). The if part follows from $\lim_{t \to \infty} f(t\xi) = \lim_{t \to \infty} \log \sum_{k} q_k e^{t \omega_k^{\top} \xi} = -\infty$.
	The only if part follows from (2).
\end{proof}
A geometric interpretation of (1) is explained as follows:
Observe that $\inf_{\xi} f(\xi) = f(\xi^*)$ if and only if $\nabla f(\xi^*) = 0$, that is,  
the origin $0$ belongs to 
the {\em gradient space} $\nabla f(\RR^d) := \{\nabla f(\xi) \mid \xi \in \RR^d \}$.
Further, $\inf_{\xi} f(\xi) > -\infty$ is equivalent to $0 \in \overline{\nabla f(\RR^d)}$, the closure of the space. 
On the other hand, $\overline{\nabla f(\RR^d)}$ is written as 
the convex hull of $\omega_k$ over $k \in [m]$.
Thus, $\xi$ in (1) represents a normal vector of a separating hyperplane 
between $0$ and $\overline{\nabla f(\RR^d)}$, which 
is a certificate of unboundedness.
See also a recent work \cite{BLNW2020} for  such an aspect 
on unconstraint geometric problem.

Specializing Lemma~\ref{lem:certificate_geometric} 
to (\ref{eqn:capacity_geometric}) 
(with $x_i = e^{\xi_i}$, $y_j = e^{\eta_j}$), we obtain
\begin{Cor}\label{cor:certificate}
	Let $A_{\rm min} := \min \{A_{ij} \mid i,j: A_{ij} > 0\}$.
	Then 
	$A$ is not approximately $(r,c)$-scalable if and only 
	if there are scaling vectors $x,y$ 
	with $\kappa(x,y) \leq A_{\rm min}$. 
	For such $x,y$, it holds  
	$$x_ky_\ell < \prod_{i} x_i^{r_i/R} \prod_j y_j^{c_j/C} \quad (\forall k,\ell: A_{k\ell} > 0).$$
\end{Cor}
The closure $\overline{\nabla \log \kappa(\RR^n \times \RR^m)}$ 
of the gradient space is the convex hull of $(e_i-r/R,e_j-c/C)$
over $i,j$ with $A_{ij} > 0$. 
The condition of Theorem~\ref{thm:RS} (1) says $0 \in \overline{\nabla \log \kappa(\RR^n \times \RR^m)}$, i.e., 
the origin $0$ satisfies all inequalities 
defining $\overline{\nabla \log \kappa(\RR^n \times \RR^m)}$.

Therefore, for a nonscalable matrix $A$, 
after finitely many Sinkhorn iterations, scaling vectors 
$x,y$ become an unbounded certificate of $\kappa(x,y) \leq A_{\rm min}$.  
One may expect that such $x,y$
can be rounded by a violating inequality, that is, a Hall blocker. 
The first algorithm in Section~\ref{sec:finding} actually does it.

\subsection{Matrix scaling as KL-divergence minimization}\label{subsec:KL}
To study asymptotic behaviors of the Sinkhorn algorithm for nonscalable matrices, 
we here consider a different optimization formulation.
In fact, as pointed out by Gietl and Reffel~\cite{GietlReffel}, 
the Sinkhorn algorithm is also viewed as alternating 
minimization of the KL-divergence between two linear subspaces of matrices.
This enables us to apply the information-geometric framework of Csisz\'ar and Tusn\'ady~\cite{CsiszarTusnady1984} to 
establish the convergence. 
In this section, following~\cite{GietlReffel} and \cite{CsiszarTusnady1984}, 
we formulate the Sinkhorn algorithm as alternating minimization 
of KL-divergence, and showed its convergence, where
we point out in Lemmas~\ref{lem:sublinear} that the proof of \cite{CsiszarTusnady1984} brings  
an explicit convergence rate.
This estimate seems not well-known but 
plays a key role in the analysis of our second algorithm.  
In this section, we do not assume $R=C$.

We consider the Sinkhorn algorithm 
in matrix update formulation $A \leftarrow C(R(A))$.
Define a sequence $\{ (M_k,N_k)\}_{k=0,1,2,\ldots}$ 
generated by the Sinkhorn algorithm:
\begin{eqnarray}
N_0 & := & C(A), \nonumber \\
M_k & := & R(N_k), \nonumber \\
N_{k+1} &:= & C(M_{k}). \label{eqn:alternating}
\end{eqnarray}	
Define the spaces ${\cal M}, {\cal R}, {\cal C}$ of  $n \times m$ matrices by
\begin{eqnarray*}
	{\cal M} &:= & \{ M \in \RR^{n \times m} \mid  M_{ij} = 0\ (i,j: A_{ij} = 0), M_{ij} \geq 0\ (i,j: A_{ij} > 0) \}, \\
	{\cal R} &:= & \{ M \in {\cal M} \mid  M{\bf 1} = r \}, \\
	{\cal C} &:= & \{ M \in {\cal M} \mid  M^{\top}{\bf 1} = c \}.
\end{eqnarray*}	
Consider the convex optimization of 
minimizing the KL-divergence from ${\cal R}$ to ${\cal C}$:
\begin{equation}\label{eqn:D(M|N)}
{\rm inf}.\quad  D(N \| M) = \sum_{i,j} N_{ij} \log N_{ij}/M_{ij} \quad {\rm s.t.}\quad  M \in {\cal R}, N \in {\cal C}.
\end{equation}
The Sinkhorn algorithm is interpreted 
as alternating minimization of this problem:\footnote{
	The same interpretation holds even if the objective 
	$D(N \| M)$ is replaced by $D(M \| N)$.}
\begin{Lem}[{see \cite{GietlReffel}}]\label{lem:alternating} 
	The following hold:
	\begin{itemize}
		\item[(1)] $\inf_{M \in {\cal R}} D(N \| M) = D(N \| R(N))$ $(N \in {\cal M})$. 
		\item[(2)] $\inf_{N \in {\cal C}} D(N \| M) = D(C(M) \| M)$ $(M \in {\cal M})$.
	\end{itemize}
\end{Lem}
The following properties link with
the convergence of the sequence $(M_k,N_k)$.
\begin{Prop}[{\cite{CsiszarTusnady1984}; see also \cite{GietlReffel}}]\label{prop:5point}
	For $M \in {\cal R}, N \in {\cal C}$, the following hold:
	\begin{description}
		\item[{\rm (3-point property)}] $ D(N \| N_k ) + D(N_k \| M_{k-1}) = D(N \| M_{k-1})$. 
		\item[{\rm (4-point property)}] $D(N \| M_{k}) \leq D(N \| N_k) + D(N \| M)$. 
	    \item[{\rm (5-point property)}] $D(N_k \| M_{k-1}) + D(N \| M_{k}) \leq D(N \| M) + D(N \| M_{k-1})$.
	\end{description}
\end{Prop}
For completeness, 
the proofs of Lemma~\ref{lem:alternating} 
and Proposition~\ref{prop:5point} are given in Appendix.

Choose an arbitrary minimizer $(\tilde M^*,\tilde N^*)$ of (\ref{eqn:D(M|N)}).
It actually exists, 
since $D$ is a lower semi-continuous convex function 
on compact set ${\cal R} \times {\cal C}$. From the 5-point property with $(M,N) = (\tilde M^*,\tilde N^*)$, we have
\begin{equation*}
D(N_k \| M_{k-1}) - D(\tilde N^* \| \tilde M^*) \leq D(\tilde N^* \| M_{k-1}) - D(\tilde N^* \| M_{k}) \quad (k=0,1,2,\ldots). 
\end{equation*}
Adding them from $k=1$ to $\ell$, we obtain
\begin{equation*}
\sum_{k=1}^\ell \left\{ D(N_k \| M_{k-1}) - D(\tilde N^* \| \tilde M^*) \right\} \leq D(\tilde N^* \|M_0) - D(\tilde N^* \| M_{\ell}). 
\end{equation*} 
Then it holds $\ell (D(M_\ell \| N_{\ell}) - D(\tilde M^* \| \tilde N^*)) \leq$ LHS since
\[
D(N_k \| M_{k-1}) \geq D(N_k \| M_{k}) \geq D(N_{k+1} \| M_{k}) \geq D(\tilde N^* \| \tilde M^*).
\] 
With $D(\tilde N^* \|M_0) - D(\tilde N^* \| M_{\ell}) = D(\tilde N^* \| (C/R)M_0) - D(\tilde N^* \| (C/R)M_{\ell}) \leq D(\tilde N^* \| (C/R)M_0)$, we obtain: 
\begin{Lem}\label{lem:sublinear} For a positive integer $\ell >0$, it holds
	\begin{equation}
	D(N_\ell \| M_\ell) - D(\tilde N^* \| \tilde M^*) \leq \frac{D(\tilde N^* \|(C/R)M_0)}{\ell}.
	\end{equation}	
\end{Lem}
Thus the sequence $D(M_k \| N_k)$ converges to 
the infimum of~(\ref{eqn:D(M|N)}).
Particularly, $\{(M_k,N_k)\}_k$ has an accumulating point $(M^*, N^*)$ 
in compact set ${\cal R} \times {\cal C}$, which is also a minimizer of (\ref{eqn:D(M|N)}).
By adding 3-point property for $k+1$ and 4-point property for $k$, 
we obtain $D(N^* \| N_{k+1}) \leq D(N^* \| N_{k}) + (D(N^*\|M^*) - D(N_{k+1}\|M_k)) \leq D(N^* \| N_{k})$.
Necessarily, $D(N^* \| N_{k})$ converges to zero. By Pinsker's inequality (Lemma~\ref{lem:Pinsker}), we have: 
\begin{Thm}[{\cite{CsiszarTusnady1984}; see \cite{GietlReffel}}]
The sequence $(M_k,N_k)$ converges to a minimizer  
$(M^*,N^*)$ of (\ref{eqn:D(M|N)}).
\end{Thm}
We call $(M^*,N^*)$ the {\em Sinkhorn limit} of $A$.
If ${\cal R} \cap {\cal C} \neq \emptyset$, 
then $D(N^* \| M^*)=0$, and $M^* = N^*$; 
this is precisely the case where $A$ is approximately $(r,c)$-scalable. 
\begin{Thm}[\cite{Bregman1967,SinkhornKnopp67}]
	Suppose that $A$ is approximately $(r,c)$-scalable. 
	$M_k$ and $N_k$ converge to an $(r,c)$-scaling $M^* = N^*$.
\end{Thm}
In particular, the approximate scalability is 
equivalent to ${\cal R} \cap {\cal C} \neq \emptyset$.
As will be seen in the beginning of the next section, 
This is a network flow feasibility condition, and Theorem~\ref{thm:RS}~(1) can be deduced from the max-flow min-cut theorem.

We continue to study the convergence of the sequence $p^{(k)}$ of column-marginals
$$
p^{(k)} := N_k {\bf 1} \quad (k=0, 1,2,\ldots).
$$
Let $p^* := N^*{\bf 1}$ of the limit of the sequence $p^{(k)}$. 
The sequence $p^{(k)}$ belongs to the following convex polytope $P := P(A,c)$:
\begin{equation} \label{eqn:P}
	P :=  \{p \in \RR^n \mid p = N{\bf 1}, N \in {\cal C} \}.
\end{equation}
\begin{Lem}\label{lem:D(r|p)}
	\begin{itemize}
		\item[(1)] $D(N\| R(N)) = D(N{\bf 1} \| r)$ for $N \in {\cal M}$.
		\item[(2)] $D(N^* \|M^*) = D(p^* \| r) = \inf_{p \in P} D(p \|r)$.
	\end{itemize}
\end{Lem}
\begin{proof}
	(1). By $R(N)_{ij} = (r_i/N {\bf 1})_i N_{ij}$, we have
	$D(N\| R(N)) = \sum_{i,j} N_{ij} \log N_{ij}/R(N_{ij}) =  \sum_{i,j} N_{ij} \log (N{\bf 1})_i/r_i = \sum_{i} (N{\bf 1})_i \log (N{\bf 1})_i/r_i = D(N{\bf 1} \| r)$.
	
	(2). 
	By Lemma~\ref{lem:alternating}~(2) and the strict convexity of $M \mapsto D(N^*\| M)$, 
	it must hold $M^* = R(N^*)$. By (1) and Lemma~\ref{lem:alternating}, 
	we have $\inf_{p \in P}D(p\|r) = \inf_{N \in {\cal C}} D(N\|R(N)) = \inf_{M \in {\cal R},N \in {\cal C}} D(N\|M) = D(N^* \| M^*) = D(N^* \| R(N^*)) 
		= D(N^*{\bf 1} \| r) = D(p^* \|r)$.
\end{proof}

Now we have the following convergence estimate of $p_k$.
\begin{Lem}\label{lem:sublinear2}
	For a positive integer $\ell$, it holds 
	\begin{equation*}
		\frac{1}{2C} \| p^{(\ell)} - p^*\|_1^2 \leq D(p^{(\ell)}\|p^*) \leq  \frac{D(N^* \|(C/R)M_0)}{\ell}.
	\end{equation*}
\end{Lem}

\begin{proof}
The first inequality is Pinsker's inequality  (Lemma~\ref{lem:Pinsker}). 	
By Lemmas~\ref{lem:sublinear} and \ref{lem:D(r|p)}, we have
\begin{equation}\label{eqn:D(p|r)}
	D(p^{(\ell)}\| r) - D(p^* \| r) \leq D(N^* \| (C/R)M_0)/\ell.
\end{equation}
The convex polytope $P$ belongs to $\{p \in \RR^n \mid p^{\top} {\bf 1}= C, p_i \geq 0 \ (i \in [n])\}$.
By Pythagoras theorem (Lemma~\ref{lem:Pythagorean}), 
we have $D(p^{(\ell)}\|p^*) \leq$ LHS in (\ref{eqn:D(p|r)}).
\end{proof}

\begin{Rem}
	Alternating minimization of this kind is studied in the literature of {\em first order methods}---a current active area of optimization; see \cite[Chapter 14]{Beck_FirstOrderMethods}.
	Since KL-divergence does not satisfy the $L$-smoothness assumption,
	we could not find from the literature results giving better estimates.
\end{Rem}

\begin{Rem}
	In the case of ${\cal R} \cap {\cal C} \neq \emptyset$, 
	the Sinkhorn algorithm is also viewed as 
	alternating minimization of minimizing KL-divergence 
	from $A$ to  ${\cal R} \cap {\cal C}$: 
	\begin{equation}
		{\rm inf.} \quad D(M \| A) \quad {\rm s.t.} \quad M \in {\cal R} \cap {\cal C}. 
	\end{equation}
    Further, it is known \cite{Csiszar1975} that the Sinkhorn limit $M^* = N^*$ is 
	the unique minimizer; see also \cite[Section 3.4]{Idel_survey}.
	Based on this formulation, Chakrabarty and Khanna~\cite{CK2021} obtained
	convergence results on the Sinkhorn algorithm.   
	They showed \cite[Lemma 2.2]{CK2021} in the proof that the KL-divergence from $M^* = N^*$ 
	to $M_k$ ($N_k$) decreases 
	by the same quantity of RHS in~(\ref{eqn:decrement}).
	Such an estimate for nonscalable setting, which we could not obtain,  
	may improve our result for finding extreme Hall blockers in Section~\ref{subsec:finding_2}. 
\end{Rem}

\section{Polymatroidal analysis on the Sinkhorn limit}\label{sec:limit}
In this section, we study the Sinkhorn limit $(M^*,N^*)$ in more detail. 
Particularly, we exhibit the block-triangular structure of $(M^*,N^*)$ (Aas~\cite{Aas2014}) in a refined form, and 
provide an explicit formula of the marginal limit $p^*$.
Although some of results (Theorems~\ref{thm:N*1} and \ref{thm:limit}) can be deduced by refining the arguments in Aas~\cite{Aas2014}, 
we here present different and self-contained proofs from polymatroidal viewpoints.
This establishes a new link with 
DM-decomposition, principle partition, and parametric stable sets in bipartite graph; 
see the last paragraph of this section.

The polytope $P$ in (\ref{eqn:P})  
is also viewed as the set 
of all vectors $p \in \RR^n$ for which 
$A$ is approximately $(p,c)$-scalable. 
That is, by Theorem~\ref{thm:RS} (1), it holds
\begin{equation}\label{eqn:P1}
P = \{ p \in \RR^n \mid p(X) + c(Y) \leq C \ (X \sqcup Y \in {\cal S}), p([n])=C \}.
\end{equation}
One can see the equivalence between  (\ref{eqn:P}) and (\ref{eqn:P1})
by a standard argument of network flow:
Regard each edge $ij$ in $G$ a directed edge from $i$ to $j$ having infinite capacity. 
Add source $s$ and sink $t$ together with directed edges $si$ $(i \in [n])$ and $jt$ $(j \in [m])$.  The capacity of $si$ is defined as $p_i$ and the capacity of $jt$ is defined as $c_i$.  Let $\vec G(A,p,c)$ denote the resulting network.
From $N$ in (\ref{eqn:P}). we obtain an $(s,t)$-flow in $\vec G(A,p,c)$ 
such that the flow-values of edges $si, ij, jt$ are 
$p_i, N_{ij}, c_j$, respectively.
In this way, $N$ can be identified with a (maximum) flow of the flow-value $C = p([n])$. 
Also a stable set $X \sqcup Y \in {\cal S}$ is identified with an $(s,t)$-cut $\{s\} \cup X \cup ([m] \setminus Y)$ with finite capacity $p([n]) - p(X) + C - c(Y)$. 
Then, by the max-flow min-cut theorem,
the condition $p(X) + c(Y) \leq C$ $(X \sqcup Y \in {\cal S})$ with $p([n])=C$ 
 is necessarily and sufficient for $p$ to be represented as $p = N{\bf 1}$.

As seen in Lemma~\ref{lem:D(r|p)}, 
the value $D(N^* \| M^*)$ is given by optimization problem:
\begin{equation}\label{eqn:polymatroid_optimization}
{\rm inf}. \quad D(p \| r) \quad {\rm s.t.} \quad p \in P.
\end{equation}  
We are going to give an explicit formula of the (unique) optimal solution $p^*$.
As mentioned after Theorem~\ref{thm:RS},   
the inequality system for $P$ can be written as
\begin{equation}\label{eqn:inequality_system}
p(X) \leq  c(\Gamma(X)) \ (X \subseteq [n]), \quad p([n]) = c (\Gamma([n])) (= C). 
\end{equation}
Observe that $X \mapsto c(\Gamma(X))$ is a monotone submodular function, 
i.e., $c(\Gamma(X))+ c(\Gamma(Y)) \geq c(\Gamma(X \cup Y))+ c(\Gamma(X \cap Y))$ for $X,Y \subseteq [n]$
and $c(\Gamma(X)) \subseteq c(\Gamma(Y))$ if $X \subseteq Y$.
Therefore, $P$ is the {\em base polytope} of the {\em polymatroid}; see \cite{FujiBook}.

Consider a map ${\cal S} \to \RR^2$ by
\begin{equation}\label{eqn:map}
X \sqcup Y \mapsto (r(X),c(Y)) \quad (X \sqcup Y \in {\cal S}),
\end{equation}
and consider the convex hull, denoted by $\Conv {\cal S}$,  
of points $(r(X),c(Y))$ over all stable sets $X \sqcup Y \in {\cal S}$.
The convex polygon $\Conv {\cal S}$ is contained by box  $[R, 0] \times [0, C]$ and
 contains $(0,0)$, $(R,0)$, $(0,C)$ as extreme points. 
A stable set $X \sqcup Y$ is called {\em extreme} 
if $(r(X),c(Y))$ is a nonzero extreme point of $\Conv {\cal S}$.  
See the left of Figure~\ref{fig:diagram}.
	\begin{figure}[t]
	\begin{center}
		\includegraphics[scale=0.45]{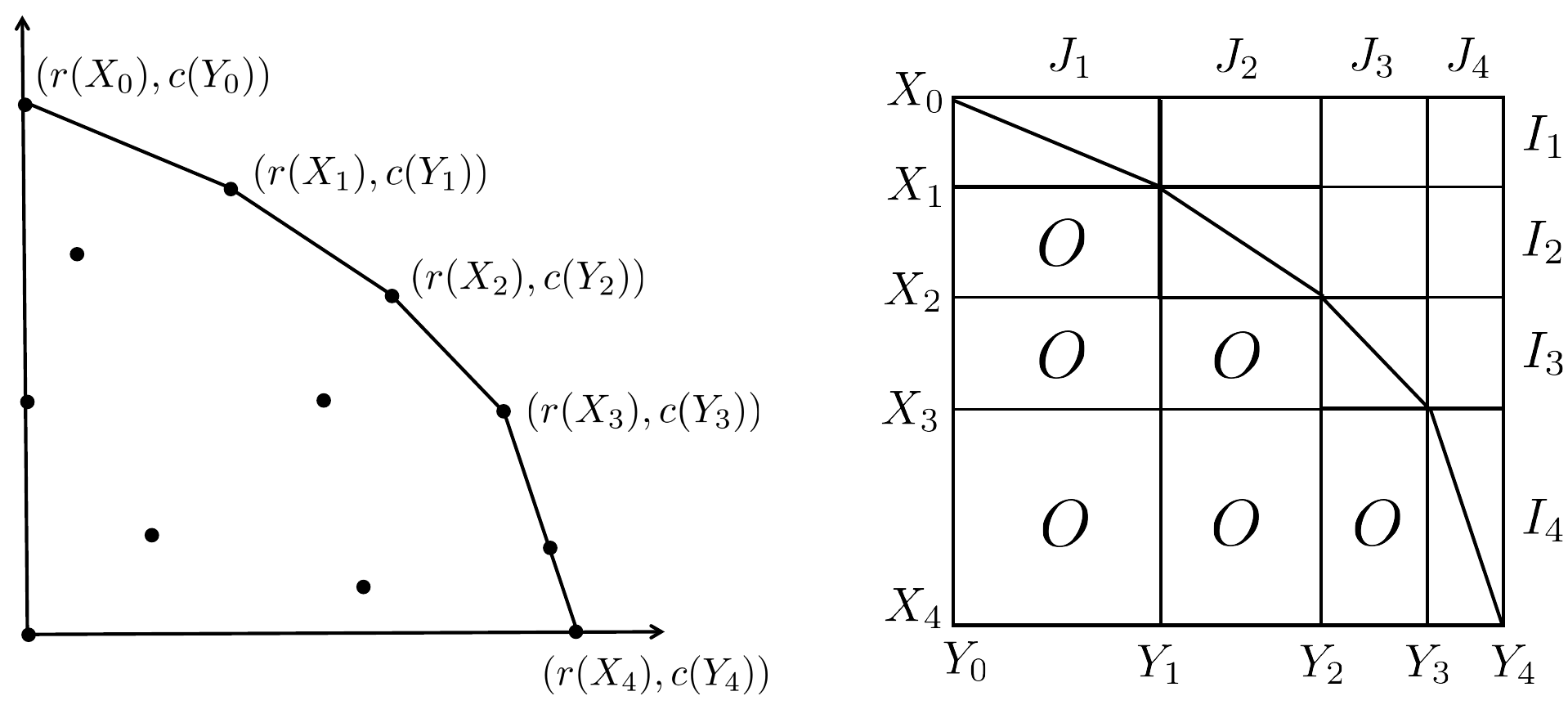}
		\caption{$\Conv {\cal S}$ in $(y,x)$-plane (left) 
			and the associated block-triangularization of $A$ (right)}
		\label{fig:diagram}
	\end{center}
\end{figure}\noindent
\begin{Lem}\label{lem:extreme}
	For extreme stable sets $X \sqcup Y$, $X' \sqcup Y'$, if  
    $r(X) \geq r(X')$ and $c(Y) \leq c(Y')$, then
	$X \supseteq X'$ and $Y \subseteq Y'$.
	In particular, the map (\ref{eqn:map}) is injective on 
	the family of extreme stable sets.
\end{Lem}
\begin{proof}
	It suffices to consider 
	the case where 
	points $(r(X),c(Y))$ and $(r(X'),c(Y'))$ are equal or adjacent extreme points.
	Observe that $(X \cap X') \sqcup (Y \cup Y')$ 
	and  $(X \cup X') \sqcup (Y \cap Y')$ are both stable, and that
	\[
	 (r(X), c(Y)) + (r(X'), c(Y')) = (r(X \cap X'), c(Y \cup Y')) 
	 + (r(X \cup X'), c(Y \cap Y')).
	\]
	This means that the point 
	of the RHS divided by $2$ 
	is equal to
	the midpoint of the edge between $(r(X),c(Y))$ and $(r(X'),c(Y'))$.
	Necessarily $(r(X), c(Y)) = (r(X \cup X'), c(Y \cap Y'))$ must hold 
	otherwise $(r(X \cap X'), c(Y \cup Y'))$ 
	or $(r(X \cup X'), c(Y \cap Y'))$ goes outside of $\Conv {\cal S}$. 
	Since $r$ and $c$ are positive vectors, 
	it must hold $X = X \cup X'$ and $Y = Y \cap Y'$, as required.
\end{proof}
Consider all extreme stable sets $X_{\kappa} \sqcup Y_{\kappa}$ $(\kappa=0,1,2,\ldots,\theta)$.
Since $c$ is a positive vector, it necessarily holds
\begin{equation}\label{eqn:[m]-Gamma(X)}
 Y_{\kappa} = [m] \setminus \Gamma(X_{\kappa}).
\end{equation} 
By Lemma~\ref{lem:extreme}, they can be ordered as
\begin{eqnarray*}
&&	[n] = X_0 \supset X_1 \supset \cdots \supset X_{\theta} = \emptyset,  \\
&&   \emptyset = Y_0 \subset Y_1 \subset \cdots \subset Y_{\theta} = [m].
\end{eqnarray*}
Let $(I_1,I_2,\ldots,I_{\theta})$ and $(J_1,J_2,\ldots,J_{\theta})$ be partitions of 
$[n]$ and $[m]$, respectively, defined by
\begin{equation}\label{eqn:partition}
I_{\kappa} := X_{\kappa-1}/X_{\kappa},\quad J_{\kappa} := Y_{\kappa}/Y_{\kappa-1} \quad (\kappa = 1,2,\ldots,\theta).
\end{equation}
Notice that $A[X,Y] = O$ if and only if $X \sqcup Y \in {\cal S}$.
We can arrange $A$ so that
these submatrices $A[I_\kappa, J_{\kappa}]$ are placed in diagonal positions and
their lower left blocks are all zero matrices. 
See the right of Figure~\ref{fig:diagram}.

Let $r^\kappa$ denote the $|I_{\kappa}|$-dimensional 
vector obtained by restricting $r$ 
to the indices in $I_{\kappa}$. 
Similarly, let $c^\kappa$ denote the restriction of $c$ 
to $J_{\kappa}$.  
Let $R_\kappa := r(I_{\kappa})$ and $C_\kappa := c(J_\kappa)$ denote the total sums of $r^\kappa$ and $c^\kappa$, respectively.
\begin{Lem}\label{lem:DM}
	\begin{itemize}
		\item[(1)] Each submatrix $A[I_{\kappa},J_{\kappa}]$ is approximately $(r^\kappa/R_\kappa, c^\kappa/C_\kappa)$-scalable.
		\item[(2)] It holds $R_1/C_1 < R_2/C_2 < \cdots < R_\theta / C_\theta$.
	\end{itemize}	
\end{Lem}
\begin{proof}
	(1). Suppose otherwise.
	Then there is a stable set $X \sqcup Y \subseteq I_{\kappa} \sqcup J_{\kappa}$ 
	such that $(1/R_{\kappa})r^\kappa(X) + (1/C_{\kappa})c^\kappa(Y) > 1$.
	Notice that
	the equation of the line through extreme points $(r(X_{\kappa-1}), c(Y_{\kappa-1}))$ 
	and $(r(X_{\kappa}), c(Y_{\kappa}))$ in the $(x,y)$-plane $\RR^2$ is 
	given by $(1/R_{\kappa})(x-r(X_{\kappa})) + (1/C_{\kappa})(y-c(Y_{\kappa-1})) = 1$. 
	Then $(X \cup X_{\kappa}) \sqcup (Y \cup Y_{\kappa-1})$ is a stable set
	in the original graph $G$ such that the point $(r(X \cup X_{\kappa}), c(Y \cup Y_{\kappa-1}))$ is outside of $\Conv {\cal S}$.
	This is a contradiction.
	
	(2). It is clear from the observation 
	that $- R_\kappa/C_\kappa$ is a slope of the edge 
	between extreme points 
	$(r(X_{\kappa-1}),c(Y_{\kappa-1}))$ and $(r(X_{\kappa}), c(Y_{\kappa}))$ in the convex polygon $\Conv {\cal S}$.  
\end{proof}

Define $\rho^* \in \RR^n$ by $\rho^*_i := (C_{\kappa}/R_{\kappa}) r_i$ 
for $i \in I_{\kappa}$. 
By an appropriate ordering, $\rho^*$ is written as
\begin{equation}\label{eqn:r*}
\rho^* = 
\left(
\begin{array}{c}
(C_1/R_1) r^1 \\
(C_2/R_2) r^2 \\
\vdots \\
(C_\theta/R_\theta) r^\theta
\end{array}
\right).
\end{equation}

\begin{Prop}\label{prop:minimizer}
	$\rho^*$ is the unique minimizer of (\ref{eqn:polymatroid_optimization}). 
\end{Prop}
\begin{proof}
We first verify that $\rho^*$ is a point in $P$. 
By Lemma~\ref{lem:DM}~(1) and (\ref{eqn:P})$(=$(\ref{eqn:P1})), for each $\kappa$, there is 
$|I_{\kappa}| \times |J_{\kappa}|$ matrix $B_{\kappa}$ 
such that $(B_{\kappa})_{ij} = 0$ for each $i,j$ with $(A[I_\kappa,J_\kappa])_{ij} = 0$, $B_{\kappa} {\bf 1} = r^{\kappa}/R_{\kappa}$, and $B_{\kappa}^{\top} {\bf 1} = c^{\kappa}/C_\kappa$. Let $B = \bigoplus_{\kappa} C_{\kappa} B_{\kappa}$ 
be the block diagonal matrix with diagonals $C_{\kappa} B_{\kappa}$. Then $B \in {\cal C}$, and $\rho^* = B{\bf 1}$.

We next verify that $\rho^*$ satisfies 
the KKT-condition for (\ref{eqn:polymatroid_optimization}).  
The gradient $\nabla D(p \| r) = (\partial D(p \| r)/ \partial p_i )_i$  at $\rho^*$ is given by
\begin{eqnarray}
&& \nabla D(\rho^* \| r) = {\bf 1} - \sum_{i} \log \frac{r_i}{\rho^*_i} {\bf 1}_{\{i\}}
= {\bf 1} - \sum_{\kappa=1}^\theta \sum_{i \in I_{\kappa}} \log \frac{R_\kappa}{C_\kappa} {\bf 1}_{\{i\}} \nonumber \\
&&= {\bf 1} - \sum_{\kappa=1}^{\theta} \log \frac{R_\kappa}{C_\kappa} 
({\bf 1}_{X_{\kappa-1}} - {\bf 1}_{X_{\kappa}}) \nonumber\\ 
&& = \left(1 -  \log \frac{R_1}{C_1}\right) {\bf 1} -  \sum_{\kappa=1}^{\theta-1} \left ( \log \frac{R_{\kappa+1}}{C_{\kappa+1}} - \log \frac{R_{\kappa}}{C_{\kappa}} \right) {\bf 1}_{X_\kappa}. \label{eqn:KKT}
\end{eqnarray}
Also $\rho^*$ satisfies 
the following inequalities from (\ref{eqn:inequality_system}) in equality:
\begin{equation*}
({\bf 1}_{X_\kappa})^{\top} \rho^* = \sum_{\lambda=\kappa+1}^{\theta} \rho^*(I_{\lambda}) 
= \sum_{\lambda=\kappa+1}^{\theta} (C_\lambda/R_{\lambda})R_\lambda 
= \sum_{\lambda = \kappa+1}^{\theta} c(J_{\lambda}) = c([m] \setminus Y_{\kappa}) = c(\Gamma(X_{\kappa})),
\end{equation*}
where we use (\ref{eqn:[m]-Gamma(X)}) in the last equality.
In particular, $\rho^*$ fulfills the equality constraint ${\bf 1}^{\top} \rho^* = C$ of $P$.
By Lemma~\ref{lem:DM} (2), it holds $\log (R_{\kappa+1}/C_{\kappa+1}) - \log(R_{\kappa}/C_{\kappa}) >0$.
Thus, (\ref{eqn:KKT}) says the KKT-condition that 
$- \nabla D(\rho^* \| r)$ is 
a nonnegative combination of ${\bf 1}_{X_\kappa}$ and $\pm {\bf 1}$. 
By strict convexity of $p \mapsto D(p \| r)$ on $P$, it is a unique minimizer.
\end{proof}
Thus we have the following.
\begin{Thm}\label{thm:N*1}
	$\rho^* = p^* (= N^*{\bf 1})$.
\end{Thm}
In particular, it holds ${\bf 1}_{I_{\kappa}}^{\top} N^*{\bf 1}_{[m] \setminus Y_{\kappa-1}} 
= p^{*}(I_{\kappa}) 
= C_{\kappa} = {\bf 1}_{[n]\setminus X_{\kappa}}^{\top} N^*{\bf 1}_{J_{\kappa}}$.
This means that all off-diagonal blocks of $N^*$ are zero, 
i.e., $N^* = \bigoplus_{\kappa} N^{*}[I_{\kappa},J_{\kappa}]$, 
where the diagonal block $N^*[I_{\kappa},J_{\kappa}]$ has marginal 
$((C_{\kappa}/R_{\kappa})r^{\kappa},c^{\kappa})$.
Accordingly, $M^* = R(N^*)$ is also a diagonal matrix $M^* = \bigoplus_{\kappa} M^{*}[I_{\kappa},J_{\kappa}]$ such that 
$M^*[I_{\kappa},J_{\kappa}]$ has marginal 
$(r^{\kappa}, (R_{\kappa}/C_{\kappa})c^{\kappa})$.

Although the preceding argument is enough for our goal (Theorem~\ref{thm:main}),
we continue to study the Sinkhorn limit $(M^*,N^*)$ to give a final form (Theorem~\ref{thm:limit}).
Extreme stable sets are obtained via the following parametric 
optimization for parameter $\lambda \in [0,1]$:
\begin{equation}\label{eqn:parametric}
{\rm max}. \quad (1- \lambda) r(X) + \lambda c(Y) \quad {\rm s.t.} \quad X \sqcup Y \in {\cal S}. 
\end{equation}
This is a minimum $(s,t)$-cut problem on the network $\vec G(A, (1-\lambda)r, \lambda c)$. 
Indeed, $X \sqcup Y$ corresponds to $(s,t)$-cut 
$\{s\} \cup X \cup ([m] \setminus Y)$ with capacity $- (1-\lambda) r(X) - \lambda c(Y) + (1-\lambda)R+ \lambda C$. 
Particularly, (\ref{eqn:parametric}) for all $\lambda$ can be efficiently solved 
by a parametric maximum flow algorithm, e.g., \cite{GGT}.

An extreme stable set is precisely 
a unique maximizer of (\ref{eqn:parametric}) for some $\lambda \in [0,1]$.
A parameter $\lambda$ is said to be {\em critical} if a minimizer is not unique.
The number of critical parameters, i.e., the number of slopes in $\Conv {\cal S}$,  
equals $\theta$. 
Let $\lambda_1, \lambda_2, \ldots \lambda_{\theta}$ be critical parameters ordered as
$0 < \lambda_1 < \lambda_2 < \cdots < \lambda_{\theta} < 1$. 
%
Let ${\cal S}_{\kappa}$ denote the family of maximizers of 
(\ref{eqn:parametric}) for $\lambda = \lambda_{\kappa}$.
Then it is not difficult to see
\begin{itemize}
	\item $X \sqcup Y, X' \sqcup Y' \in {\cal S}_{\kappa} \Longrightarrow (X \cup X') \sqcup (Y \cap Y'), (X \cap X') \sqcup (Y \cup Y')  \in {\cal S}_{\kappa}$.
	\item $X_{\kappa-1} \sqcup  Y_{\kappa-1}, X_{\kappa} \sqcup  Y_{\kappa} \in {\cal S}_{\kappa}$.
	\item $X \sqcup Y \in {\cal S}_{\kappa} \Longrightarrow X_{\kappa-1} \supseteq X \supseteq X_{\kappa}, Y_{\kappa-1} \subseteq Y \subseteq Y_{\kappa}$.
\end{itemize}
Namely, ${\cal S}_{\kappa}$ admits a distributive lattice structure, 
where $(X \sqcup Y) \wedge (X' \sqcup Y') := (X \cup X') \sqcup (Y \cap Y')$
and $(X \sqcup Y) \vee (X' \sqcup Y') := (X \cap X') \sqcup (Y \cup Y')$.
This is essentially a well-known fact that 
the family of minimum $(s,t)$-cuts forms a distributive lattice.

Choose a maximal set (chain) $\{ X_{\kappa,\alpha} \sqcup Y_{\kappa,\alpha}\}_{\alpha=0,1,2,\ldots,\alpha_{\kappa}} \subseteq {\cal S}_{\kappa}$ with the property:
\begin{eqnarray*}
	&&	X_{\kappa-1} = X_{\kappa,0} \supset X_{\kappa,1} \supset \cdots \supset X_{\kappa, \alpha_{\kappa}} = X_{\kappa},  \\
	&&  Y_{\kappa-1} = Y_{\kappa,0} \subset Y_{\kappa,1} \subset \cdots \subset Y_{\kappa, \alpha_{\kappa}} = Y_{\kappa}.
\end{eqnarray*}
This is also efficiently obtained from the residual network of 
$\vec G(A, (1-\lambda_{\kappa}) r, \lambda_{\kappa} c)$ with respect to any maximum flow.
Accordingly, define $I_{\kappa,\alpha}$ and $J_{\kappa,\beta}$ by
\begin{equation}\label{eqn:partition2}
I_{\kappa,\alpha} := X_{\kappa,\alpha-1}/X_{\kappa,\alpha},\quad J_{\kappa,\alpha} := Y_{\kappa,\alpha}/Y_{\kappa,\alpha-1} \quad (\alpha = 1,2,\ldots,\alpha_\kappa).
\end{equation}
Now $\{I_{\kappa,\alpha}\}_{\kappa,\alpha}$ and $\{J_{\kappa,\alpha}\}_{\kappa,\alpha}$  
are partitions of $[n]$ and $[m]$ refining $\{I_{\kappa}\}_{\kappa}$ and $\{J_{\kappa}\}_{\kappa}$, respectively. 
Let $r^{\kappa,\alpha}$ and $c^{\kappa,\alpha}$ denote the restrictions of $r$ and $c$ to $I_{\kappa,\alpha}$ and $J_{\kappa,\alpha}$, respectively.
Again, we arrange $A$ so that 
$A[I_{\kappa,\alpha}, J_{\kappa,\alpha}]$ are diagonal blocks and their lower left blocks are zero. This is a finer block-triangularization of $A$.

The point $(r(X_{\kappa,\alpha}), c(Y_{\kappa, \alpha}))$ lies on 
the segment between extreme points 
$(r(X_{\kappa-1}), c(Y_{\kappa-1}))$ and $(r(X_{\kappa}), c(Y_{\kappa}))$ in $\Conv {\cal S}$.
By the maximality of the chain, 
there is no stable set $X \sqcup Y$ in $I_{\kappa,\alpha} \sqcup J_{\kappa,\alpha}$
such that the stable set $(X_{\kappa,\alpha} \cup X) \sqcup (Y_{\kappa,\alpha-1} \cup Y)$ 
is mapped to the segment between 
$(r(X_{\kappa,\alpha-1}),c(Y_{\kappa,\alpha-1}))$ and $(r(X_{\kappa,\alpha}),c(Y_{\kappa,\alpha}))$.
This implies the condition of Theorem~\ref{thm:RS}~(2):
\begin{Lem} Each submatrix $A[I_{\kappa,\alpha}, J_{\kappa,\alpha}]$ is 
	$(r^{\kappa,\alpha}/R_{\kappa}, c^{\kappa,\alpha}/C_{\kappa})$-scalable.  
\end{Lem}
Notice that $C_{\kappa}/R_{\kappa} = c(J_{\kappa,\alpha})/r(I_{\kappa,\alpha})$ for each $\alpha$. Again, all off-diagonal blocks in the limit 
$M^*[I_{\kappa},J_{\kappa}]$ must be zero, as in the argument after Theorem~\ref{thm:N*1}.
Thus we have the following.
\begin{Thm}\label{thm:limit}
	The Sinkhorn limit $(M^*,N^*)$ is given by
	\begin{equation}
	M^* = \bigoplus_{\kappa \in [\theta], \alpha \in [\alpha_{\kappa}]} R_{\kappa} B_{\kappa,\alpha}, \quad 
		N^* = \bigoplus_{\kappa \in [\theta], \alpha \in [\alpha_{\kappa}]} C_{\kappa} B_{\kappa,\alpha},
	\end{equation}
	where $B_{\kappa,\alpha}$ denotes the (unique) $(r^{\kappa,\alpha}/R_{\kappa}, c^{\kappa,\alpha}/C_{\kappa})$-scaling of $A[I_{\kappa,\alpha}, J_{\kappa,\alpha}]$.
\end{Thm}

\paragraph{Relation to principal partition and DM-decomposition.}
The problem (\ref{eqn:parametric}) 
is a parametric submodular optimization, since it 
is also written as minimization of $X \mapsto t c(\Gamma(X)) -  (1-t)r(X)$.
Such a parametric problem has been studied in a general framework of 
the {\em principal partition} of polymatroids~\cite{Fujishige2009,TomizawaFujishige}.
The above nested distributive lattice structure and the associated partition 
are generalized to this setting, 
which has a number of applications in combinatorial optimization; 
see the above surveys.

In the case of $(r,c) = ({\bf 1}, {\bf 1})$, 
the block-triangularization of $A$ obtained from the chain 
$X_{\kappa,\alpha} \sqcup Y_{\kappa,\alpha}$
is a refinement of {\em Dulmage-Mendelsohn decomposition (DM-decomposition)}~\cite{DM1}; 
see \cite[Section 4.3]{LovaszPlummer} and \cite[Section 2.2.3]{MurotaMatrixMatroid}.
The DM-decomposition uses a chain of maximizers of (\ref{eqn:parametric}) 
only for $\lambda = 1/2$.
This refined DM-decomposition considering all parameters is due to  
N. Tomizawa (unpublished 1977); see \cite{TomizawaFujishige}.

\section{Finding Hall blockers by Sinkhorn iteration}\label{sec:finding}

Let $G = (U \sqcup V, E)$ be a bipartite graph 
with color classes $U,V$ and edge set $E$.
We assume that there is no isolated node. 
Suppose that $n:= |U| = |V|$.
The following is well-known:
	\begin{description}
		\item[{\rm (Hall's theorem)}]
	$G$ has a perfect matching if and only 
	if $|X| \leq |\Gamma(X)|$ for all $X \subseteq [n]$.
	\item [{\rm (K\H{o}nig-Egerv\'{a}ry theorem)}] The maximum cardinality of a matching in $G$
	is equal to the minimum of $|U| -|X| + |\Gamma(X)|$ over $X \subseteq U$.
	\end{description}

Let $(r,c) := ({\bf 1}, {\bf 1}) \in \RR^{n} \times \RR^{n}$.
Define $n \times n$ matrix $A_G$ by $(A_G)_{ij} = 1$ if $ij \in E$ and zero otherwise, 
where rows and columns of $A_G$
are indexed by $U$ and $V$, respectively.
In this setting, the Hall condition for $G$ is nothing but 
the approximate scalability condition in Theorem~\ref{thm:RS}~(1) for $A_G$ with  $(r,c) := ({\bf 1}, {\bf 1})$. 

In this section, we show that a polynomial number of iterations of the Sinkhorn algorithm identifies a Hall blocker if $G$ has no perfect matching ($A_G$ is not scalable).
To make clarify the roles of parameters,  instead of $A_G$
we apply the presented algorithms to 
a general nonnegative matrix $A$ with $G(A) = G$.
In addition to $A_{\rm min}$ in Corollary~\ref{cor:certificate}, 
define $A_{\rm max}$ 
by the maximum of entries $A_{ij}$ of $A$
and $m$ by the number of nonzero elements of $A$ (the number of edges of $G$).

\subsection{A Hall blocker from scaling vectors}
Our first algorithm is based on 
the geometric programming analysis in Section~\ref{subsec:geometric}.
Here the Sinkhorn iteration is described in update of scaling vectors $x,y$. 
\begin{description}
	\item [Algorithm: Sinkhorn \& Sorting Scaling Vectors]
	\item [0:]  Let $x  =y := {\bf 1}$, and let $y \leftarrow C(x)$.
	\item [1:] Repeat row- and column-normalization $\ell$ times:
	$$ x \leftarrow R(y), \quad y \leftarrow C(x).$$
	\item[2:]  
	Sort $x,y$ as $x_{i_1} \geq x_{i_2} \geq \cdots \geq x_{i_n}$,  
	$y_{j_1} \leq y_{j_2} \leq \cdots \leq y_{j_n}$.
	Choose index $k^*$ such that  $(x_{i_{k^*}}y_{j_{k^*}})^n \geq 
	x_1x_2\cdots x_n y_1 y_2 \cdots y_n$.
	Output $\{i_1,\ldots,i_{k^*}\}$.
\end{description}

\begin{Thm}\label{thm:main0}
	Suppose that $G$ has no perfect matching.
	For $\ell \geq (1/4) n^2 \log m A_{\rm max}/A_{\rm min}$, the output is a Hall blocker.
\end{Thm}
The recovering procedure from $x,y$ to a Hall blocker in step 2
is essentially the same as in \cite[Lemma 3.3]{FGS2022}.
The new point here is the iteration bound.
We remark that this bound (almost) matches the one by \cite{AWR2017,CK2021} for scalable case.  
Namely, 
$O((n^2/\epsilon^2) \log m A_{\rm max}/A_{\rm min})$ iterations 
find either a Hall blocker or an approximate $({\bf 1}, {\bf 1})$-scaling with $\ell_1$-error $\epsilon \in (0,2)$.

\begin{proof}
	Notice $\kappa ({\bf 1},{\bf 1}) = \|A\|_1 \leq m A_{\rm max}$.
	We compute the decrement of $\kappa$ in one iteration of step 1.
	Let $x' := R(y)$, $y' := C(x')$,  $B := (A_{ij}x_iy_j)$, and $B' := (A_{ij}x'_iy_j)$. 
	By Lemmas~\ref{lem:Pinsker}, \ref{lem:l_1}, and \ref{lem:decrement} with $B^{\top}{\bf1} = {\bf 1}$ and $B'{\bf1} = {\bf 1}$, we have
	\begin{eqnarray*}
	&& \log \kappa (x',y') - \log \kappa (x,y)  = - D({\bf 1} \| B{\bf 1})/n -  D({\bf 1} \| (B')^{\top}{\bf 1})/n \\
	&& \leq - \frac{1}{2n^2} \left\{ \| B{\bf 1} - {\bf 1} \|_1^2+ \| (B')^{\top}{\bf 1} - {\bf 1} \|_1^2 \right\} \\
	&& \leq - \frac{1}{2n^2} \left\{ ( \| B{\bf1} - {\bf 1} \|_1 +  \| B^{\top}{\bf1} - {\bf 1} \|_1)^2 
	+  (\| B'{\bf1} - {\bf 1} \|_1 +  \| (B')^{\top}{\bf1} - {\bf 1} \|_1)^2   \right\} \\
	&& \leq  - \frac{4}{n^2}.
	\end{eqnarray*}
	Therefore, after $\ell \geq (1/4) n^2 \log m A_{\rm max}/A_{\rm min}$ iterations, we have
	$\log \kappa (x,y) \leq \log \kappa ({\bf 1},{\bf 1}) - 4\ell/n^2 \leq 
	\log m A_{\rm max} -  \log mA_{\rm max}/A_{\rm min} \leq \log A_{\rm min}$.

	By Corollary~\ref{cor:certificate}, 
	it holds $x_iy_j < (x_1\cdots x_ny_1\cdots y_n)^{1/n}$ for each edge $ij$ in $G$. 
	Therefore, $G$ has no edge between $i_k$ and $j_{\ell}$ 
	for $k \leq k^*$, $\ell \geq k^*$, 
	since $x_{i_k}y_{j_\ell} \geq x_{i_{k^*}}y_{j_{k^*}} \geq (x_1\cdots x_ny_1\cdots y_n)^{1/n}$.
	Thus $\Gamma(\{i_1,i_2,\ldots,i_{k^*}\}) \subseteq \{j_1,j_2,\ldots,j_{k^*-1} \}$. 
	That is, the output $\{i_1,i_2,\ldots,i_{k^*}\}$ is a Hall blocker.
\end{proof}
\begin{Rem}\label{rem:general_marginals}
	The algorithm and its analysis can be adapted 
	for integer marginals $r,c$ with $R=C$. 
	In the step 2, choose indices $k^*,\ell^*$ such that 
	$$
	\left(\sum_{k=1}^{k^*-1}r(i_k),\sum_{k=1}^{k^*}r(i_k)\right) \cap 	\left(\sum_{\ell=1}^{\ell^*-1}c(j_{\ell}), \sum_{\ell=1}^{\ell^*}c(j_{\ell})\right) \neq \emptyset, \quad (x_{i_{k^*}}y_{j_{\ell^*}})^R \geq \prod_{i} x_i^{r_i} \prod_j y_j^{c_j}.
	$$
	Then $\{i_1,\ldots,i_{k^*}\}$ is a Hall blocker.
	The number of iterations is $(1/4)R^2 \log m A_{\rm max} /A_{\rm min}$.
	We omit the details. 
\end{Rem}

\subsection{Extreme Hall blockers from marginals}\label{subsec:finding_2}
Next we show that further iterations identify all extreme Hall blockers from marginal vector $p = A{\bf 1}$.
Here a subset $X$ is called {\em extreme} if 
$X \sqcup ([n] \setminus \Gamma(X))$ is an extreme stable set in Section~\ref{sec:limit}.
If $G$ has no perfect matching, then
extreme subsets other than trivial ones $\emptyset, [n]$ are all Hall blockers.
Particularly, they include a Hall blocker $X$ with maximum $|X| - |\Gamma(X)|$ and minimum (maximum) $|X|$.
The analysis is based on the KL-formulation 
in Section~\ref{subsec:KL} and Section~\ref{sec:limit}, 
where the Sinkhorn iteration is performed in the matrix update formulation. 
\begin{description}
	\item[Algorithm: Sinkhorn \& Sorting Marginals]
	\item[0:] $A \leftarrow C(A)$.
	\item[1:] Repeat row- and column-normalization $\ell$ times:
	$$A \leftarrow C(R(A)).$$
	\item[2:] Let $p := A{\bf 1}$, and sort $p$ as 
	$p_{i_1} \leq p_{i_2} \leq \cdots \leq p_{i_{n}}$.
	\item[3:] Output $U_k := \{i_1, i_{2},\ldots,i_k \}$ for $k=0, 1,2,\ldots,n$, where $U_0 := \emptyset$. 
\end{description}
\begin{Thm}\label{thm:main}
	Suppose that $G$ has no perfect matching.
	For $\ell > 16n^6 \log n A_{\rm max}/A_{\rm min}$, the output contains all extreme Hall blockers.
\end{Thm}

We start the proof. 
According to (\ref{eqn:alternating}), 
define the sequence $(M_k, N_k)$ with the initial point $N_0 := C(A)$.
Consider the family $\{X_{\kappa} \sqcup Y_{\kappa} \}_{\kappa=0,1,\ldots,\theta}$ of extreme stable sets and the associated
partitions $(I_1,I_2,\ldots,I_{\theta})$ and $(J_1,J_2,\ldots,J_{\theta})$ 
in (\ref{eqn:partition}). 
Then the limit $p^* = N^*{\bf 1}$ is given by 
\[
p^{*\kappa} = \frac{|J_{\kappa}|}{|I_{\kappa}|} {\bf 1} \quad (\kappa = 1,2,\ldots,\theta).
\]
Observe from $|J_{1}|/|I_{1}| > |J_{2}|/|I_{2}| > \cdots > |J_{\theta}|/|I_{\theta}|$ 
that sorting $p^*$ recovers all $X_{\kappa}$.

In the step 2 of the algorithm, the matrix $A$ equals $N_\ell$.
Let $p := A{\bf 1}$. 
\begin{Lem}\label{lem:1/2n^2}
	If $\|p^* - p \|_{\infty} < 1/(2n^2)$, then the output contains $X_{\kappa}$ 
	for all $\kappa = 0,1,2,\ldots, \theta$.
\end{Lem}
\begin{proof}
	If $p_i > p_j$ holds for all $(i,j) \in I_{\kappa} \times I_{\lambda}$ with $\kappa < \lambda$, 
	then the output contains all $X_{\kappa}$.
	Since $|J_{\kappa}|/|I_{\kappa}| - |J_{\lambda}|/|I_{\lambda}| \geq 1/n^2$ for $\kappa < \lambda$, 
	for $i \in I_{\kappa}, j \in I_{\lambda}$, we have
	$p_i - p_j > p_i^* - 1/(2n^2) - (p_j^* + 1/(2n^2)) > 1/n^2 - 1/n^2 =0$, as required. 
\end{proof}

\begin{Lem}\label{lem:nlogn}
	$D(N^* \| M_0) \leq 2 n \log n A_{\rm max}/A_{\rm min}$.
\end{Lem}
\begin{proof}
	For  $q_j := \sum_{k} A_{kj}$ and $p_i := \sum_j C(A)_{ij}$, we have
	\begin{eqnarray*} 
		D(N^* \| M_0) &=& D(N^* \| R(C(A))) = \sum_{ij \in E} N^{*}_{ij} \log N^*_{ij}/(A_{ij}/p_iq_j) \\
		&=& \sum_{ij \in E} N^{*}_{ij} \log N^*_{ij}p_i + \sum_{ij \in E} N^{*}_{ij} \log \frac{\sum_k A_{kj}}{A_{ij}} \\
		&\leq & n \log n + n \log n A_{\rm max}/A_{\rm min} \leq 2 n \log n A_{\rm max}/A_{\rm min}, 
\end{eqnarray*}
where we use $N_{ij}^* \leq 1$, $p_i \leq n$, and $\sum_{ij} N_{ij}^* = n$.
\end{proof}
%

\begin{proof}[Proof of Theorem~\ref{thm:main}]
By Lemma~\ref{lem:1/2n^2}, 
it suffices to show that $\|p^* - p \|_{\infty} < 1/(2n^2)$ holds in step 2.
For $\ell > 16 n^6 \log n A_{\rm max}/A_{\rm min}$,
by Lemmas~\ref{lem:sublinear2} and \ref{lem:nlogn}, we have
\begin{eqnarray*}
	&& \|p_{\ell} - p^*\|_{\infty} \leq \|p_\ell - p^* \|_1 
	\leq \sqrt{\frac{4n^2 \log n A_{\rm max}/A_{\rm min}}{\ell}} < \frac{1}{2n^2}.
\end{eqnarray*}
\end{proof}
\begin{Rem}
	The algorithm is also adapted for general integer marginals $r,c$ with $R=C$.
	In step 2, define $\tilde p$ by $\tilde p_i := p_i/r_i$, and sort $\tilde p$ instead of $p$.
	The iteration bound is obtained simply by replacing $n^6$ with $R^6$.
\end{Rem}

\section*{Acknowledgments}
We thank the referees for helpful comments, 
Satoru Fujishige for bibliographical information on principal partition,
and Masahito Hayashi for discussion on alternating minimization.
The first author was supported by Grant-in-Aid for JSPS Research Fellow, Grant No. JP19J22605, Japan.
The second author was supported by JST PRESTO Grant Number JPMJPR192A, Japan.

\normalsize
\appendix
\section{Appendix}
\paragraph{Proof of Lemma~\ref{lem:alternating}.}
	(1). Let  $p_i := \sum_{j} N_{ij}$. For any $M \in {\cal R}$, we have
	\begin{eqnarray*}
		&& D(N\|M)= \sum_{i,j} N_{ij} \log N_{ij}/M_{ij} \geq \sum_{i,j} N_{ij} 
		\log \sum_{j} N_{ij}/ \sum_j M_{ij} \\ 
		&& = \sum_{i,j} N_{ij} \log p_i / r_i = \sum_{i,j} N_{ij} \log N_{ij}/R(N)_{ij} = D(N \| R(N)), 
	\end{eqnarray*}	
	where we use the convexity of $D$ 
	and $R(N)_{ij} =(r_i/p_i) N_{ij}$.
	
	(2). Let $q_j := \sum_{i} M_{ij}$. For any $N \in {\cal C}$, we have
	\begin{eqnarray*}
		&& D(N \| M) \geq \sum_{i,j} N_{ij} 
		\log \sum_{i} N_{ij}/ \sum_i M_{ij} \\
		&& = \sum_j c_j \log c_j/q_j = \sum_{i,j} C(M)_{ij} \log C(M)_{ij}/M_{ij} = D(C(M) \| M), 
	\end{eqnarray*} 
	where $C(M)_{ij} =(c_j/q_j) M_{ij}$.

\paragraph{Proof of Proposition~\ref{prop:5point}.}
	The 5-point property is obtained by adding the 3-point and 4-point properties. 
     The 3-point property follows from
	\begin{eqnarray*}
		&& D(N \| M_{k-1}) = D(N \| N_k) + \sum_{i,j} N_{ij} \log \frac{(N_k)_{ij}}{(M_{k-1})_{ij}} = D(N \| N_k) + D(N_k \| M_{k-1}), \\
	\end{eqnarray*}	
where the second equality follows from the fact that
    $(N_k)_{ij}/(M_{k-1})_{ij} = ((M_{k-1})^{\top}{\bf 1})_j/c_j$ does not depends on $i$, and 
    $\sum_i N_{ij} = \sum_{i} (N_k)_{ij} = c_j$.
    
	The 4-point property follows from  
	\begin{eqnarray*}
	&& D(N \| M_{k}) - D(N \| N_k) = \sum_{i, j} N_{ij} \log \frac{(N_k)_{ij}}{(M_{k})_{ij}} = \sum_{i} \sum_j N_{ij} \log \frac{(N_k)_{ij}}{R(N_k)_{ij}} \\
	&& = \sum_{i} (N{\bf 1})_i \log \frac{(N_k {\bf 1})_i}{(M{\bf 1})_i} = D(N{\bf 1} \| M{\bf 1}) - D(N{\bf 1} \| N_k{\bf 1})  \\
	&& \leq D(N{\bf 1} \| M{\bf 1}) \leq D(N \| M),
	\end{eqnarray*}
	where the final inequality is the log-sum inequality.
\end{document}